\documentclass{article}
\usepackage{amssymb}
\usepackage{amsmath}
\usepackage{amssymb}
\usepackage{amsmath}
\usepackage{amsthm}
\usepackage{bm,bbm}
\usepackage{geometry}
\usepackage{graphicx}
\usepackage{caption}
\usepackage{enumerate}
\usepackage[numbers]{natbib}
\usepackage{algorithm}
\usepackage{algpseudocode}
\usepackage[colorlinks,CJKbookmarks=true,bookmarksnumbered,linkcolor=black,citecolor=black,plainpages=true,pdfstartview=FitH]{hyperref}
\theoremstyle{plain}
\newtheorem{theorem}{Theorem}

\newtheorem{lemma}[theorem]{Lemma}

\theoremstyle{definition}
\newtheorem{definition}[theorem]{Definition}

\theoremstyle{remark}

\usepackage[utf8]{inputenc}

\newcommand{\prob}{\mathbb{P}}
\newcommand{\E}{\mathbb{E}}

\newcommand{\indicator}{\mathbb{I}}

\newcommand{\real}{\mathbb{R}}

\title{Smoothed Analysis of Inconsistent A*}
\author{Zhiyang Chen, Hailong Yao}
\date{}

\begin{document}

\maketitle
\begin{abstract}
    The A* search is a fundamental path-finding algorithm in artificial intelligence. While admissible and consistent heuristics guarantee efficient performance by expanding each state at most once, modern search applications frequently employ powerful but inconsistent heuristics derived from machine learning, randomized evaluations, etc. A long-standing theoretical barrier to using these inconsistent heuristics is the risk of catastrophic node re-expansion, which yields a worst-case exponential time complexity of $\Omega(2^n)$. However, empirical observations contradict this pessimistic bound, demonstrating that inconsistent A* operates highly efficiently in practice. 

    To bridge this significant gap between theory and practice, this paper presents the first smoothed analysis of the A* algorithm using inconsistent heuristics. We model typical real-world noise by applying slight random perturbations to the edge weights of worst-case search graphs. Our main result proves that the expected smoothed time complexity of inconsistent A* is bounded by a polynomial, specifically a total iteration number of $O(n^2 m \kappa)$, where $n$ is the number of nodes, $m$ is the number of edges, and $\kappa$ controls the scale of random perturbations. Furthermore, we also show that this result naturally extends to the functionally equivalent problem of Dijkstra's algorithm on negative-weight graphs.
\end{abstract}

\section{Introduction}
The A* search algorithm stands as one of the most fundamental and widely utilized path-finding algorithms in artificial intelligence. Its efficiency is largely driven by its heuristic function, $h(\cdot)$, which estimates the cost from an arbitrary node to the goal node. It is a well-established property that if $h(\cdot)$ is {\em admissible} (i.e., it never overestimates the true cost), A* is guaranteed to return an optimal path. Furthermore, if $h(\cdot)$ is {\em consistent}, (i.e., the difference between the heuristic values of two adjacent nodes does not exceed the weight of the edge connecting both nodes), A* is guaranteed to expand each state at most once. Thus, using an admissible and consistent heuristic, A* finds the shortest path in at most $n-1$ iterations for $n$-node graphs.

However, many powerful heuristics derived in practice are inconsistent (sometimes even not admissible). These include heuristics generated by machine learning, randomized evaluations, and compressed pattern databases~\cite{felner11inconsistent,sakaue22sample}. The theoretical barrier to deploying inconsistent heuristics lies in the risk of node re-expansion. When a shorter path to a previously expanded node is discovered, A* must reopen that node. In a seminal result, \citet{martelli77on} demonstrated that this reopening process can cascade catastrophically, yielding a worst-case time complexity of $\Omega(2^n)$. Consequently, the prevailing theoretical consensus has historically cautioned against the use of inconsistent heuristics.

In contrast, \citet{felner11inconsistent} has argued that, in practice, inconsistent A* is generally not as inefficient as the worst-case complexity illustrates. Moreover, inconsistent heuristics are easy to create, adding a diversity of heuristic values, which may lead to a reduction of re-expansions. \citet{felner11inconsistent} suggests that ``the use of inconsistent heuristics will become an accepted and powerful tool in the development of high-performance search algorithms.''

Therefore, a large gap exists between theory and practice. It is natural to ask the following question: \textbf{Why do A* algorithms with inconsistent heuristics avoid exponential blow-up in practice?}

Standard algorithmic analysis frameworks struggle to answer this question. Worst-case analysis is overly pessimistic, considering highly artificial graph structures that are fragile and unlikely to occur in real-world domains. Conversely, average-case analysis typically requires rigid, often unrealistic assumptions about the probability distribution of the input graphs and heuristic values, failing to capture the structure of actual problem instances.

To bridge this gap, we turn to the framework of {\em smoothed analysis}, introduced by \citet{spielman04smoothed}, to explain the practical efficiency of the Simplex algorithm. Smoothed analysis elegantly interpolates between worst-case and average-case analysis. It measures the performance of an algorithm under slight random perturbations of worst-case inputs. If an algorithm exhibits polynomial smoothed complexity, it implies that the pathological instances inducing worst-case exponential behavior are mathematically fragile, which rarely exist due to the random noise in practice.

In this paper, we provide the first rigorous smoothed analysis of the A* algorithm using inconsistent heuristics. We assume the edge costs of the search graph are independently perturbed by small random noises. In expectation, the running time of the A* algorithm for smoothed instances will only be a polynomial in the graph size $n$ and the scale of randomness.

\subsection{Our contributions}
\noindent\textbf{Main results}. The main contribution of this work is the following theorem.
\begin{theorem}[Informally, Theorem~\ref{thm:main}]
    The smoothed time complexity of inconsistent A* is bounded from above by $\text{poly}(n,\kappa)$, where $n$ is the number of nodes in the search graph, and $\kappa$ is the scale of random perturbations.
\end{theorem}

We also empirically verify the theoretical analysis by numerical experiments. We find that, the worst-case exponential-time instances for A* constructed by \citet{martelli77on} are highly sensitive in practice. The total number of iterations greatly reduces as we increase the perturbation on the instances.

Moreover, it is widely known that there is an equivalence between the A* algorithm and the Dijkstra algorithm on graphs with negative weights. Our analysis also works for Dijkstra. We show that, for a shortest path instance with no negative cycles, if the edge weights are randomly perturbed, Dijkstra finds the shortest path in polynomial time.

\noindent\textbf{Our techniques}. The main difficulty in proving the smoothed complexity is to find enough independent randomness among the intermediate states of inconsistent A*. Previous work on smoothed analysis of heuristic local search algorithms (e.g., \citet{englert16smoothed}) generally uses union bounds over all possible solution improvements to induce a polynomial smoothed complexity bound. However, this does not work for inconsistent A* since there could be an exponential number of path updates. The main proof technique in our analysis is to establish a surprising connection between inconsistent A* and Pareto optimization (a.k.a. multi-objective optimization). Then, we could apply the ``winner gap'' technique by \citet{beier04typical} for the smoothed analysis of Pareto frontiers. We make various modifications to this technique to adapt to our problem.

\subsection{Related work}
\noindent\textbf{Theoretical analysis of A*}. \citet{felner11inconsistent} has shown that the total number of iterations of inconsistent A* is bounded from above by $O(n^2 W)$, where $W$ is the maximum of edge weights. However, their analysis only works for integer weights. Why inconsistent A* with real weights is efficient remains a significant theoretical question.

\citet{eden22embeddings} studies the average-case running time of A* with heuristics generated by embedding and labeling schemes. However, their analysis focuses on consistent heuristics. \citet{sakaue22sample} studies the sample complexity of learning an arbitrary heuristic for an A* instance distribution. In their analysis, since the learned heuristic could be inconsistent and inadmissible, they bound the sample complexity by counting all possible paths from the start node to other nodes, which leads to exponential bounds. This indicates that bounding the time complexity of inconsistent A* could be of interest for learning-theoretic problems.

\noindent\textbf{Smoothed analysis of algorithms}. Originally proposed by \citet{spielman04smoothed} to study the time complexity of the Simplex method for linear programming, smoothed analysis is widely used to explain the practical efficiencies of a variety of algorithms, including local search~\cite{englert16smoothed,chen20smoothed}, algorithm configuration~\cite{balcan18dispersion,chen25learning}, online algorithms~\cite{coester25smoothed}, and machine learning~\cite{arthur09kmeans}. We refer readers to \citet{Roughgarden21beyond} for a comprehensive introduction.

\section{Preliminaries}
\noindent\textbf{Problem formulation}. We use $G=(V,E)$ to denote a weighted graph with $|V|=n$ nodes and $|E|=m$ edges. Each edge $(u,v)\in E$ is associated with a real weight $w_{(u,v)}\in [0,1]$. Without loss of generality, we assume that edge weights are bounded, since we can always normalize the weights by the maximum. In our analysis, we consider a directed graph, but our approach also works for undirected graphs. For simplicity of analysis, we assume the graph is simple, i.e., there is at most one edge from node $u$ to $v$ for any pair $(u,v)$, and no self-loops are allowed. Let $s\in V$ denote the start node and $t\in V$ denote the goal node. Moreover, we have a heuristic function $h:V\to [0,+\infty)$ that estimates the shortest path from each node to the goal. The target of the A* algorithm is to find a path $s\to v_1\to v_2\to\dots\to v_K\to t$ that minimizes the total weight $w_{(s,v_1)}+\sum_{i=1}^{K-1}w_{(v_i,v_{i+1})}+w_{(v_K,t)}$.

\noindent\textbf{Inconsistent A*}. The A* algorithm is illustrated in Algorithm~\ref{alg:a-star}. In each iteration, the algorithm finds an active node $v$ in $\text{OPEN}$ with the minimum $f(v)=g(v)+h(v)$, and updates the $g$ value of adjacent nodes if a shorter path is discovered using $v$. Then we add the updated nodes to $\text{OPEN}$ for further updates.

Note that since we consider inconsistent heuristics, Algorithm~\ref{alg:a-star} incorporates the reopening process. If the heuristic is not admissible, Algorithm~\ref{alg:a-star} cannot find the shortest path even with reopening. In this case, we could remove lines 4--6 of Algorithm~\ref{alg:a-star}, and run the algorithm until convergence. This guarantees optimality, but leads to a longer running time. In fact, our analysis does not rely on the admissibility condition. This means that it also works for this special case, although it is rarely used in practice.

\noindent\textbf{Smoothed analysis model}. In smoothed analysis, we assume the problem instance is randomly perturbed by a small noise. In our model, the weight of each edge $w_{(u,v)}$ is not a fixed real value, but a probability distribution supported on $[0,1]$. Given a parameter $\kappa$, the probability density function of each edge weight $w_{(u,v)}$ is a bounded function $p_{(u,v)}:[0,1]\to [0,\kappa]$. The input of the algorithm is obtained by independently sampling each edge weight $w_e$ based on the density function $p_e$. We say that such instances are {\em $\kappa$-smoothed}. Let $T(I)$ denote the running time of the algorithm on instance $I$, let $I$ denote an input, and let $\tilde{I}$ denote a $\kappa$-smoothed problem instance. The smoothed time complexity of an algorithm is defined by $$\max_{\tilde{I}}\mathop{\E}_{I\sim\tilde{I}}[T(I)],$$ where $\tilde{I}$ is taken over all $\kappa$-smoothed instances with $n$ nodes. Note that our model only perturbs the edge weights, but the heuristic values could be arbitrary.

For example, if $\kappa=1$, each edge weight is a uniform distribution over $[0,1]$, and our model reduces to average-case analysis. If $\kappa\to +\infty$, our model converges to worst-case analysis.

In our model, the heuristic function $h:V\to [0,+\infty)$ is a fixed vector. It may depend on the graph structure and edge weights, but it is independent of the random perturbations. Therefore, given a $\kappa$-smoothed instance, we can synthesize the heuristic based on the instance (e.g., the distribution of edge weights), but it should be generated before sampling the weights.

\noindent\textbf{On tie-breaking}. Note that in the smoothed analysis setting, for any two paths that differ by at least one node, with probability 1, the total lengths (and the $f$-values as well) of these two paths are different. Thus, in line 3 of Algorithm~\ref{alg:a-star}, we break ties arbitrarily without affecting our analysis.

\begin{algorithm}[htb]
\caption{Inconsistent A* with reopening.}
\label{alg:a-star}
\begin{algorithmic}[1] 
\Require A directed weighted graph $G=(V,E)$ with $|V|=n$ nodes and $|E|=m$ edges, where each edge $e=(u,v)\in E$ is associated with a real weight $w_e\in [0,1]$; A heuristic function $h:V\to[0,+\infty)$; The start node $s\in V$ and the goal node $t\in V$.
\Ensure A simple path $v_1,\dots,v_k$ where $v_1=s$ and $v_k=t$.
\State Let $\text{OPEN}=\{s\},\text{CLOSED}=\emptyset$ and $g(s)=0$;
\While {$\text{OPEN}\ne\emptyset$}
    \State Let $v\leftarrow \arg\min_{v\in\text{OPEN}}g(v)+h(v)$;
    \If {$v=t$}
        \State \Return the $s$-$t$ path by tracing $t\leftarrow p(t)\leftarrow p(p(t))\leftarrow\dots\leftarrow s$;
    \EndIf
    \For {each node $v^\prime$ s.t. $(v,v^\prime)\in E$}
        \State Let $g^\prime\leftarrow g(v)+w_{(v,v^\prime)}$;
        \If {$v^\prime\notin\text{OPEN}\cup\text{CLOSED}$}
            \State Let $g(v^\prime)\leftarrow g^\prime, p(v^\prime)\leftarrow v$, and $\text{OPEN}\leftarrow\text{OPEN}\cup \{v^\prime\}$;
        \ElsIf {$v^\prime\in\text{OPEN}$ and $g^\prime<g(v^\prime)$}
            \State Let $g(v^\prime)\leftarrow g^\prime, p(v^\prime)\leftarrow v$;
        \ElsIf {$v^\prime\in\text{CLOSED}$ and $g^\prime<g(v^\prime)$}
            \State Let $g(v^\prime)\leftarrow g^\prime, p(v^\prime)\leftarrow v$;
            \State Move $v^\prime$ from $\text{CLOSED}$ to $\text{OPEN}$;
        \EndIf
    \EndFor
    \State Move $v$ from $\text{OPEN}$ to $\text{CLOSED}$;
\EndWhile
\end{algorithmic} 
\end{algorithm}

\section{Smoothed Analysis}
Our main result is the following theorem.
\begin{theorem}\label{thm:main}
    The expected number of iterations of the A* algorithm for a $\kappa$-smoothed instance is bounded from above by $O(n^2m\kappa)$.
\end{theorem}

We first introduce some notation for the proof of Theorem~\ref{thm:main}.

\noindent\textbf{Notation}. In the A* algorithm, $g(u)$ is used to denote the path length from the start node to any node $u$. This depends on the current state of the algorithm. In our analysis, to avoid this dependency, with a little abuse of notation, for any path $P$, we use $g(P)$ to denote the length of $P$. Moreover, all paths in our analysis are simple paths unless otherwise stated.

\subsection{Direct and indirect updates}
Before diving into the proof of Theorem~\ref{thm:main}, we define direct and indirect updates in Algorithm~\ref{alg:a-star}.

For each node $v\in V$, the A* algorithm maintains two state variables: $g(v)$, the cost of the currently known shortest path from $s$ to $v$; and $p(v)$, a parent pointer indicating the immediate ancestor of $v$ in $T$. By tracing the parent pointers, the A* algorithm implicitly maintains a path tree, denoted by $T$, that originates at $s$. The path from $s$ to each node on $T$ represents the current path found by the A* algorithm. We call such paths {\em intermediate paths}.

In lines 8--16 of Algorithm~\ref{alg:a-star}, we select a pair of nodes $(v,v^\prime)$, and update $g(v^\prime)\leftarrow\min(g(v^\prime),g(v)+w_{(v,v^\prime)})$. Depending on whether the topology of the path tree $T$ is altered, we formalize these updates into two distinct categories: Direct and indirect updates. See Figure~\ref{fig:update} for an example.

\noindent\textbf{Direct updates}. A direct update occurs when the algorithm discovers a structurally novel path from $s$ to $v$ that is strictly superior to its previous intermediate path. This kind of update alters the topological structure of the path tree $T$.

Suppose the algorithm is expanding a node $v\in V$ and evaluating its outgoing edge $(v,v^\prime) \in E$. A direct update is triggered if and only if the path traversing through $v$ offers a lower cost than the currently recorded cost $g(v^\prime)$. Therefore, before the update, we have $p(v^\prime)\ne v$. After the update, the edge $(p(v^\prime), v^\prime)$ is pruned from $T$, and the new edge $(v, v^\prime)$ is added into $T$.

\noindent\textbf{Indirect updates}. An indirect update occurs when a node $v$ experiences a reduction in its path cost $g(v)$ strictly as a byproduct of a cost reduction in its current ancestor, without any alteration to the path tree's topology. Thus, the intermediate path of $v$ remains unchanged.

Let $v^\prime\in V$ be a node that has already been updated, and let $v=p(v^\prime)$ be its current parent in $T$. If the path cost $g(v)$ undergoes an update (i.e., $g(v)$ is decreased) due to a direct or indirect update upstream, the cost of the path from $s$ to $v^\prime$ is inherently reduced, despite the local edge $(v, v^\prime)$ in $T$ remaining unchanged. 
\begin{figure}
    \centering
    \includegraphics[width=0.5\linewidth]{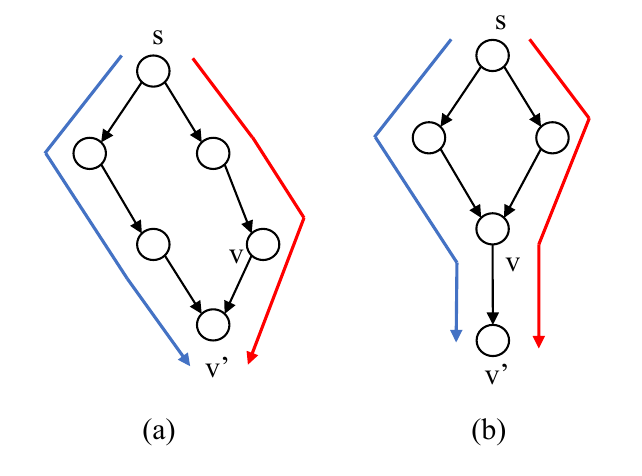}
    \caption{(a) Direct updates; (b) Indirect updates. Blue paths denote original paths, and red paths denote updated paths.}
    \label{fig:update}
\end{figure}

\begin{lemma}\label{lmm:indirect}
    For any A* instance, let $U_{\mathbf{dir}}$ and $U_{\mathbf{ind}}$ denote the number of direct and indirect updates, respectively. We have $$U_{\mathbf{ind}}\le (n-1)\cdot U_{\mathbf{dir}}.$$
\end{lemma}
\begin{proof}
    We define an update event as the event that the recorded path cost $g(v)$ of any node $v\in V$ is strictly decreased. The set of all update events $\mathcal{E}$ during the entire execution can be partitioned into two disjoint subsets: the set of direct updates $\mathcal{E}_d$ and the set of indirect updates $\mathcal{E}_i$, such that $U_{\mathbf{dir}}=|\mathcal{E}_d|$ and $U_{\mathbf{ind}}=|\mathcal{E}_i|$.

    By the definition of an indirect update, an event $e \in \mathcal{E}_i$ occurring at a node $v$ strictly preserves its parent pointer $p(v)$. Therefore, this reduction in $g(v)$ must be a mathematical consequence of a strictly decreasing update event $e^\prime$ at its parent node $u=p(v)$. We denote this causal trigger relationship as $e^\prime \to e$.

    Since the search structure $T$ maintained by A* is a directed tree rooted at the start node $s$ (which is never updated), and all update events strictly decrease $g$-values, any causal chain of updates $e \gets e^\prime \gets e^{\prime\prime} \gets \dots$ must be finite and cycle-free. 

    Consequently, tracing any indirect update event $e \in \mathcal{E}_i$ backward along this causal chain must ultimately terminate at an update event that is not an indirect update. This terminating root event must be a direct update $e_d\in \mathcal{E}_d$ occurring at some ancestor node $r$. This establishes a surjective mapping $f: \mathcal{E}_i \to \mathcal{E}_d$, meaning every indirect update is uniquely rooted in exactly one direct update.

    To prove the lemma, it suffices to bound the number of indirect updates causally generated by a single direct update. Let $e_d\in\mathcal{E}_d$ be a fixed direct update event occurring at node $r$. We define its causal progeny $S(e_d) = \{e\in \mathcal{E}_i\,|\,f(e)=e_d\}$.

    When $e_d$ occurs, it establishes a newly reduced cost $g(r)$ and a new parent $p(r)$. This specific cost reduction propagates downstream along the directed edges of $T$. Recall that the causal chain is cycle-free. There is no node that will be indirectly updated twice by a fixed direct update. Since the maximum number of nodes in any subtree of $T$ is strictly bounded by $n-1$, the maximum number of indirect updates causally triggered by $e_d$ is at most $n-1$. Thus, we have $$|S(e_d)|\le n-1,\,\forall\, e_d\in\mathcal{E}_d,$$ which proves this lemma.
\end{proof}
Therefore, in the following analysis, it suffices to bound $U_{\mathbf{dir}}$ to prove the desired result, since each iteration is triggered by at least one previous update on some node.

\subsection{Pareto conditions of node expansion}
The inconsistent A* algorithm may iterate more than $n$ times, since some nodes may be expanded more than once. The key idea of analyzing the time complexity of inconsistent A* is to characterize the conditions of node re-expansion.

Suppose node $v^\prime$ is re-expanded. Originally, the A* algorithm finds a path from $s$ to $v^\prime$. We denote it by $P_{\text{old}}$. Later, an edge relaxation on some edge $(v,v^\prime)\in E$ finds a new path $s\to\dots\to v\to v^\prime$, namely $P_{\text{new}}$, such that its total weight is smaller than $P_{\text{old}}$.

If the heuristic function $h(\cdot)$ is admissible and consistent, since $g(P_{\text{old}})>g(P_{\text{new}})$, the A* algorithm will find $P_{\text{new}}$ instead of $P_{\text{old}}$, so that re-expansion does not happen. However, if $h(\cdot)$ is inconsistent, we can manipulate $h(\cdot)$ so that $P_{\text{new}}$ has not been discovered yet when $P_{\text{old}}$ is found. Concretely, let $P_{\text{new}}$ be a path $s\to\dots\to u\to\dots\to v^\prime\to v$, where $h(u)$ is extremely large such that $u$ is in the $\text{OPEN}$ list, waiting to be expanded, while $P_{\text{old}}$ has been discovered.

Formally, we define the bottleneck comparison between paths.
\begin{definition}[Bottleneck comparison]\label{def:bottleneck}
    Fix a node $v\in V$. Let $P$ and $P^\prime$ be two different paths from $s$ to $v$, where $P=s\to u_1\to \dots \to u_k\to u_{k+1} \to \dots\to  u_K\to v$, and $P^\prime=s\to u_1\to \dots \to u_k\to u^\prime_{k+1} \to \dots\to u^\prime_{K^\prime}\to v$, with the first $k$ edges being the same ($k$ could be 0), i.e., the longest common prefix of $P$ and $P^\prime$. We say $P\prec P^\prime$ if \begin{equation*}
        \max_{i\in \{k+1,\dots,K\}} g(P(u_i))+h(u_i)<\max_{i\in\{k+1,\dots,K^\prime\}} g(P^\prime(u^\prime_i))+h(u^\prime_i),
    \end{equation*} where $P(u_i)$ is the prefix of $P$ from $s$ to $u_i$ and $P^\prime(u^\prime_i)$ is the prefix of $P^\prime$ from $s$ to $u^\prime_i$.
\end{definition}
Note that since $P$ differs $P^\prime$ by at least one edge, in the smoothed analysis setting, with probability 1, $$\max_{i\in \{k+1,\dots,K\}} g(P(u_i))+h(u_i)\ne\max_{i\in\{k+1,\dots,K^\prime\}} g(P^\prime(u^\prime_i))+h(u^\prime_i).$$ Therefore, in our definition, we only consider the strict comparison. Moreover, $\prec$ is a total order.
\begin{lemma}
    For any paths $P_1,P_2,P_3$, we have $$(P_1\prec P_2) \land (P_2\prec P_3)\implies P_1\prec P_3.$$
\end{lemma}
\begin{proof}
    Let $K_1,K_2,K_3$ denote the number of nodes in $P_1,P_2,P_3$ except the first and the last nodes, respectively. Let $k_1$ denote the length of the longest common prefix (LCP) of $P_1$ and $P_2$, $k_2$ denote the LCP of $P_2$ and $P_3$, and $k_3$ denote the LCP of $P_1$ and $P_3$. For simplicity, we use $f_j(u)=g(P_j(u))+h(u)$ to denote the $f$-value on path $P_j$. We use $u_i$ to denote the $i$-th node on the paths $P_1,P_2$ or $P_3$ (we omit the path index for simplicity).
    
    Altogether, there are three cases:

    \noindent\textbf{Case 1}: $k_1<k_2$. Since $P_2$ and $P_3$ share more common prefix edges than $P_1$ and $P_2$, it must be $k_3=k_1$. By $P_1\prec P_2$, we have $$\max_{i\in\{k_1+1,\dots,K_1\}}f_1(u_i)<\max_{i\in\{k_1+1,\dots,K_2\}}f_2(u_i).$$ By $P_2\prec P_3$, we have \begin{align*}
        \max_{i\in\{k_1+1,\dots,K_2\}}f_2(u_i)&=\max\left\{\max_{i\in\{k_1+1,\dots,k_2\}}f_2(u_i),\max_{i\in\{k_2+1,\dots,K_2\}}f_2(u_i)\right\}\\
        &\le \max\left\{\max_{i\in\{k_1+1,\dots,k_2\}}f_3(u_i),\max_{i\in\{k_2+1,\dots,K_3\}}f_3(u_i)\right\}\\
        &=\max_{i\in\{k_3+1,\dots,K_3\}}f_3(u_i),
    \end{align*} which implies $P_1\prec P_3$.

    \noindent\textbf{Case 2}: $k_1>k_2$. In this case, it must be $k_3=k_2$. Using a similar argument to Case 1, we have $P_1\prec P_3$.

    \noindent\textbf{Case 3}: $k_1=k_2$. In this case, $P_1$ and $P_3$ share at least the first $k_1$ edges. We have $k_3\ge k_1=k_2$. By $P_1\prec P_2$ and $P_2\prec P_3$, we have $$\max_{i\in\{k_1+1,\dots,K_1\}} f_1(u_i)<\max_{i\in\{k_1+1,\dots,K_2\}} f_2(u_i)<\max_{i\in\{k_1+1,\dots,K_3\}} f_3(u_i).$$ Since the inequality is strict, eliminating the common nodes of $P_1$ and $P_3$ does not affect the inequality, i.e., $$\max_{i\in\{k_3+1,\dots,K_1\}} f_1(u_i)<\max_{i\in\{k_3+1,\dots,K_3\}} f_3(u_i),$$ which implies $P_1\prec P_3$.
\end{proof}

Intuitively, Definition~\ref{def:bottleneck} compares two paths using the largest $f$-value among the path nodes (excluding the common part). These nodes are the bottlenecks of the paths. The A* algorithm finds these paths only if the minimum $f$-value in the $\text{OPEN}$ list exceeds the bottleneck value.

\begin{lemma}\label{lmm:pareto}
    In a smoothed A* instance, suppose $P$ is an intermediate path for some node $v$. With probability 1, for any other path $P^\prime$ from $s$ to $v$, we have $g(P)<g(P^\prime)$ or $P\prec P^\prime$.
\end{lemma}
\begin{proof}
    Proof by contradiction. Suppose there is a path $P^\prime$ that satisfies both $g(P^\prime)<g(P)$ and $P^\prime\prec P$. Following the notation in Definition~\ref{def:bottleneck}, let $l=\arg\max_{i\in\{k+1,\dots,K\}} g(P(u_i))+h(u_i)$. Since $P^\prime\prec P$, for any $i\in\{k+1,\dots,K^\prime\}$, we have $g(P^\prime(u_i^\prime))+h(u_i^\prime)<g(P(u_l))+h(u_l)$. Therefore, before $u_l$ is expanded, the path $P^\prime$ will be discovered. Since $g(P^\prime)<g(P)$, $P$ will not become an intermediate path. This leads to a contradiction.
\end{proof}
Lemma~\ref{lmm:pareto} indicates that there is a Pareto-optimality condition for intermediate paths. For any node $v$, let $P_1,P_2,\dots,P_N$ denote the sequence of intermediate paths of $v$, where $g(P_1)>g(P_2)>\dots>g(P_N)$. We have $P_1\prec P_2\prec\dots\prec P_N$. All paths are Pareto-minimum with respect to $g(\cdot)$ and the bottleneck comparison. See Figure~\ref{fig:pareto} for an illustration.
\begin{figure}
    \centering
    \includegraphics[width=0.6\linewidth]{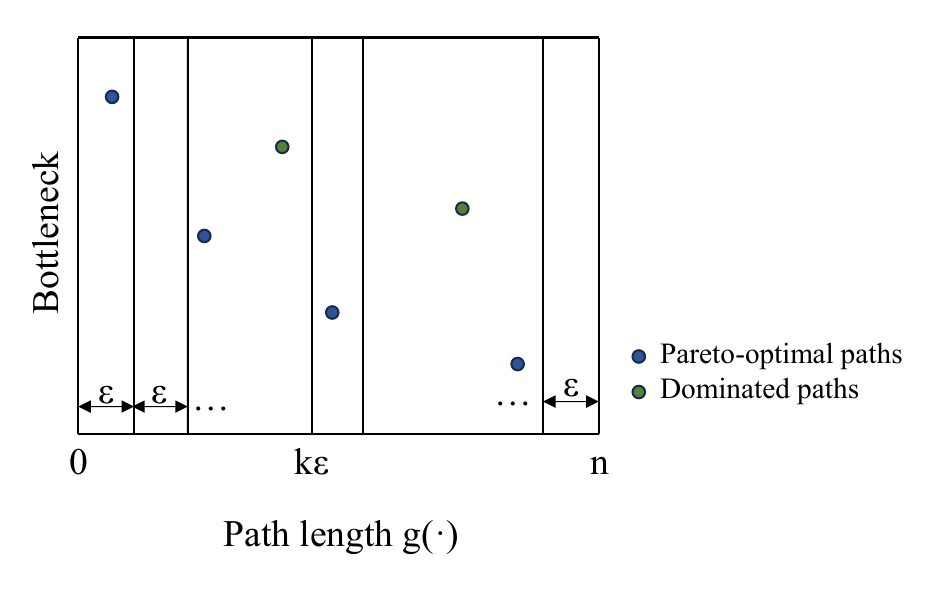}
    \caption{Pareto-optimal conditions for node expansions.}
    \label{fig:pareto}
\end{figure}

\subsection{Smoothed analysis of direct updates}
The main idea of proving Theorem~\ref{thm:main} is to bound the number of intermediate paths for each node. Summing for all nodes yields the upper bound of the total iterations of Algorithm~\ref{alg:a-star}.

Fix a node $v$. As previously discussed, we can bound the number of intermediate paths of $v$ by bounding the size of the Pareto set. Let $P_1,P_2,\dots,P_N$ denote the sequence of intermediate paths of $v$, where $g(P_1)>g(P_2)>\dots>g(P_N)$. Each pair of adjacent paths $(P_i,P_{i+1})$ denotes an update. If the last edges of $P_i$ and $P_{i+1}$ are identical, this is an indirect update. Otherwise, this is a direct update. Due to Lemma~\ref{lmm:indirect}, we only consider the number of direct updates. We use the ``loser gap'' technique to prove a smoothed upper bound, but with a few modifications to adapt our problem.

Recall that the edge weights are in $[0,1]$. For any path $P$, we have $g(P)\in [0,n]$. We partition $[0,n]$ into disjoint intervals $[0,\varepsilon), [\varepsilon,2\varepsilon),\dots,[n-\varepsilon,n]$, each with length $\varepsilon$, where $\varepsilon>0$ is a parameter to be determined later. Consider all simple paths from $s$ to $v$. Each path $P$ has a total weight $g(P)$, which lies in some interval $k\varepsilon<g(P)<(k+1)\varepsilon$. (Recall that with probability 1, the total weight of $g(P)$ will not be in a set of measure zero. Thus, it suffices to consider open intervals.) We will choose $\varepsilon$ so small that there is at most one path weight $g(P)$ that appears in each interval. See Figure~\ref{fig:pareto} for an example.

We define a few random variables.
\begin{itemize}
    \item Let $P_{v,e,k}$ for node $v\in V$, edge $e=(u,v)\in E$ and integer $k$ be a simple path from $s$ to $v$ that satisfies the following: (1) $P_{v,e,k}$ does not include edge $e$; (2) The length $g(P_{v,e,k})<k\varepsilon$; (3) Among all paths satisfying (1) and (2), $P_{v,e,k}$ is the path with the minimum bottleneck.
    \item Let $Q_{v,e,k}$ be a path from $s$ to $v$ that satisfies the following: (1) $Q_{v,e,k}$ includes edge $e$; (2) $Q_{v,e,k}\prec P_{v,e,k}$; (3) Among all paths satisfying (1) and (2), $Q_{v,e,k}$ is the path with the minimum length $g(\cdot)$.
    \item Let $G_{v,e,k}=g(Q_{v,e,k})$ be the length of $Q_{v,e,k}$.
\end{itemize}
\begin{lemma}\label{lmm:pareto-size}
    For an A* instance, suppose for each node $v$ and each interval $I_k=(k\varepsilon,(k+1)\varepsilon)$, (1) there is at most one path from $s$ to $v$, such that the path length lies in $I_k$; (2) there is no path length that lies on the terminals of any interval $I_k$. Let $U_{\mathbf{dir}}$ denote the number of direct updates. We have \begin{equation*}
        U_{\mathbf{dir}}\le \sum_{v\ne s}\sum_{k=0}^{\lfloor n/\varepsilon\rfloor-1}\sum_{(u,v)\in E}\indicator[G_{v,e,k}\in (k\varepsilon,(k+1)\varepsilon)].
    \end{equation*}
\end{lemma}
\begin{proof}
    The proof idea is to show that, for each direct update, there exists a unique term $G_{v,e,k}$ such that $G_{v,e,k}\in (k\varepsilon,(k+1)\varepsilon)$ holds.

    Fix a Pareto-optimal intermediate path $P$ from $s$ to $v$ where the last edge is $e=(v^\prime\to v)$. Suppose $g(P)\in (k\varepsilon,(k+1)\varepsilon)$. Since there is at most one path in each interval, the length of any path from $s$ to $v$ that is smaller than $g(P)$ does not exceed $k\varepsilon$. If there is a direct update on $P$, it must be another Pareto-optimal intermediate path $P^\prime$ such that (1) $P^\prime$ does not include edge $e$; (2) $g(P^\prime)<k\varepsilon$; (3) the bottleneck of $P^\prime$ is as small as possible. By definition, we find that $P^\prime$ is exactly $P_{v,e,k}$.

    Moreover, since $P$ and $P^\prime$ are two adjacent Pareto-minimum paths, $P$ must satisfies the following: (1) $P\prec P^\prime$; (2) $P$ has the minimum length among all paths that contain edge $e$ and have a smaller bottleneck than $P^\prime$. By definition, we find that $P$ is exactly $Q_{v,e,k}$. Thus, we have $G_{v,e,k}\in (k\varepsilon,(k+1)\varepsilon)$. Since each interval contains at most one intermediate path, we charge each Pareto-minimum path to a unique $G_{v,e,k}$, which leads to the desired bound.
\end{proof}

\begin{lemma}\label{lmm:interval-density}
    For a $\kappa$-smoothed A* instance, for any $v,e,k$, we have $$\prob[G_{v,e,k}\in (k\varepsilon,(k+1)\varepsilon)]\le \varepsilon\kappa.$$
\end{lemma}
\begin{proof}
    We use the principle of deferred decisions to prove this lemma. Fix the randomness of all edges except $w_e$, the weight of $e$. Now, since $P_{v,e,k}$ does not contain edge $e$, $P_{v,e,k}$ is deterministic. Consider the path length of $Q_{v,e,k}$. $Q_{v,e,k}$ satisfies three conditions: (1) $Q_{v,e,k}$ contains edge $e$; (2) $Q_{v,e,k}\prec P_{v,e,k}$; (3) the path length is minimal. Since the definition of bottleneck comparison does not depend on the last edge of the path, all candidate paths of $Q_{v,e,k}$ that satisfy (1) and (2) are deterministic. Moreover, the lengths of all these paths can be represented as the form of $\text{a deterministic term}+w_e$. Therefore, the minimum path $Q_{v,e,k}$ is also deterministic, independent of the value of $w_e$.

    Therefore, conditioning on all the randomness of $w_e$, the probability of $G_{v,e,k}\in (k\varepsilon,(k+1)\varepsilon)$ is equivalent to the probability that $w_e$ lies in an interval of length $\varepsilon$. Since $w_e$ is $\kappa$-bounded, we have $\prob[G_{v,e,k}\in (k\varepsilon,(k+1)\varepsilon)]\le \varepsilon\kappa$.
\end{proof}

\subsection{Putting things together}
Now, we are ready to prove Theorem~\ref{thm:main}.
\begin{proof}[Proof of Theorem~\ref{thm:main}]
    Let $E$ denote the event that there is an interval $I_k=(k\varepsilon,(k+1)\varepsilon)$, such that there are two or more paths from $s$ to some node $v$, where the path lengths both lie in $I_k$. Suppose there are $M$ paths from $s$ to any other node $v$ in total. We have \begin{align*}
        M\le \sum_{i=1}^{n-1} \binom{n-1}{i}i!\le (2n)^{n+1}.
    \end{align*} By a union bound over all intervals and all pairs of paths, we have $\prob[E]\le \binom{M}{2}\cdot\frac{n}{\varepsilon}\cdot (\varepsilon\kappa)^2=(2n)^{2n+3}\varepsilon\kappa^2$, since for any two different paths, there is at least one non-common edge, which leads to the fact that both paths lie in an interval of length $\varepsilon$ with probability at most $(\kappa\varepsilon)^2$. Note that $\varepsilon$ can be arbitrarily small, the probability of $E$ can be neglected (e.g., by setting $\varepsilon<\frac{1}{\kappa^2(2n)^{4n+4}}$).
    
    Let $U_{\mathbf{dir}}$ and $U_{\mathbf{ind}}$ denote the number of direct and indirect updates, respectively. The total number of iterations of A* is equivalent to the total number of updates. Therefore, the total iteration number $U$ satisfies $$\E[U]=\E[U_{\mathbf{dir}}+U_{\mathbf{ind}}]\le n\cdot \E[U_{\mathbf{dir}}]$$ by Lemma~\ref{lmm:indirect}. Now, we apply the smoothed analysis of the Pareto-minimum paths. If $E$ does not happen, we can simply apply Lemma~\ref{lmm:pareto-size}. However, if $E$ happens, the number of updates can be at most $M$. Concretely,
    \begin{align*}
        \E[U]&\le n\cdot \E[U_{\mathbf{dir}}]\\
        &\le n\cdot\left(\sum_{v\ne s}\sum_{k=0}^{\lfloor n/\varepsilon\rfloor-1}\sum_{(u,v)\in E}\indicator[G_{v,e,k}\in (k\varepsilon,(k+1)\varepsilon)]+M\cdot \prob[E]\right)\\
        &\le n\cdot \left(\sum_{v\ne s}\sum_{k=0}^{\lfloor n/\varepsilon\rfloor-1}\sum_{(u,v)\in E}\varepsilon\kappa+M\cdot \frac{1}{(2n)^{2n+1}}\right)\\
        &\le n\cdot \left(\sum_{v\ne s}\sum_{(u,v)\in E}n\kappa+o(1/n)\right)\\
        &\le O(n^2m\kappa),
    \end{align*} by Lemma~\ref{lmm:pareto-size} and Lemma~\ref{lmm:interval-density}.
\end{proof}

\section{Connections to Negative-Weight Dijkstra}
It is well-known that inconsistent A* is equivalent to the Dijkstra algorithm on negative-weight graphs without negative cycles~\cite{martelli77on}. Indeed, the original heuristic function $h(\cdot)$ can be seen as a potential function. For each edge $(u,v)\in E$, we can transform the edge weight into $w^\prime_{(u,v)}=w_{(u,v)}+h(v)-h(u)$. Recall that the Dijkstra algorithm maintains a path length vector $\text{dist}:V\to\real$ that represents the current path lengths from the start node $s$ to any node. In each iteration, the algorithm finds an active node with the minimal path length, and relaxes adjacent nodes. Due to the existence of negative-weight edges, nodes may be updated for multiple times, which leads to an exponential worst-case complexity.

\begin{algorithm}[htb]
\caption{Negative-weight Dijkstra's algorithm with reopening.}
\label{alg:dijkstra}
\begin{algorithmic}[1] 
\Require A directed weighted graph $G=(V,E)$ with $|V|=n$ nodes and $|E|=m$ edges, where each edge $e=(u,v)\in E$ is associated with a real weight $w_e\in [-1,1]$; The start node $s\in V$.
\Ensure The shortest paths from $s$ to other nodes.
\State Let $\text{OPEN}=\{s\},\text{CLOSED}=\emptyset$ and $\text{dist}(s)=0$;
\While {$\text{OPEN}\ne\emptyset$}
    \State Let $v\leftarrow \arg\min_{v\in\text{OPEN}}\text{dist}(v)$;
    \For {each node $v^\prime$ s.t. $(v,v^\prime)\in E$}
        \State Let $d^\prime\leftarrow \text{dist}(v)+w_{(v,v^\prime)}$;
        \If {$v^\prime\notin\text{OPEN}\cup\text{CLOSED}$}
            \State Let $\text{dist}(v^\prime)\leftarrow d^\prime, p(v^\prime)\leftarrow v$, and $\text{OPEN}\leftarrow\text{OPEN}\cup \{v^\prime\}$;
        \ElsIf {$v^\prime\in\text{OPEN}$ and $d^\prime<\text{dist}(v^\prime)$}
            \State Let $\text{dist}(v^\prime)\leftarrow d^\prime, p(v^\prime)\leftarrow v$;
        \ElsIf {$v^\prime\in\text{CLOSED}$ and $d^\prime<\text{dist}(v^\prime)$}
            \State Let $\text{dist}(v^\prime)\leftarrow d^\prime, p(v^\prime)\leftarrow v$;
            \State Move $v^\prime$ from $\text{CLOSED}$ to $\text{OPEN}$;
        \EndIf
    \EndFor
    \State Move $v$ from $\text{OPEN}$ to $\text{CLOSED}$;
\EndWhile
\State \Return $\text{dist}$;
\end{algorithmic} 
\end{algorithm}
Now, we study the smoothed complexity of negative-weight Dijkstra. Again, given a graph with $n$ nodes and $m$ edges, we assume the weight of each edge $e$ is independently sampled from a distribution with density function $p_e:[-1,1]\to[0,\kappa]$. We aim to find the shortest path lengths from a start node $s$ to other nodes. Such instances are called $\kappa$-smoothed.
\begin{theorem}
    Given a $\kappa$-smoothed shortest path instance where the instance has no negative cycle with probability 1, the expected number of iterations of Algorithm~\ref{alg:dijkstra} is bounded from above by $O(n^2m\kappa)$.
\end{theorem}

The proof of this theorem is essentially identical to Theorem~\ref{thm:main}. The only difference is in defining a slightly different version of bottleneck comparison. Given two different paths $P$ and $P^\prime$ from $s$ to some node $v$, where $P=s\to u_1\to \dots \to u_k\to u_{k+1} \to \dots\to  u_K\to v$, and $P^\prime=s\to u_1\to \dots \to u_k\to u^\prime_{k+1} \to \dots\to u^\prime_{K^\prime}\to v$, with the first $k$ edges being the same ($k$ could be 0), we say $P\prec P^\prime$ if \begin{equation*}
    \max_{i\in \{k+1,\dots,K\}} w(P(u_i))<\max_{i\in\{k+1,\dots,K^\prime\}} w(P^\prime(u^\prime_i)),
\end{equation*} where $P(u_i)$ is the prefix of $P$ from $s$ to $u_i$, $P^\prime(u^\prime_i)$ is the prefix of $P^\prime$ from $s$ to $u^\prime_i$, and $w(\cdot)$ is the total length of a path. The rest of the proof follows the same argument as Theorem~\ref{thm:main}.

\section{Numerical Experiments}
We perform numerical experiments on the inconsistent A* instances constructed by \citet{martelli77on} to show that the exponential-time behavior of inconsistent A* is highly sensitive to random perturbations, thus verifying our smoothed analysis result.

\begin{figure}
    \centering
    \includegraphics[width=0.6\linewidth]{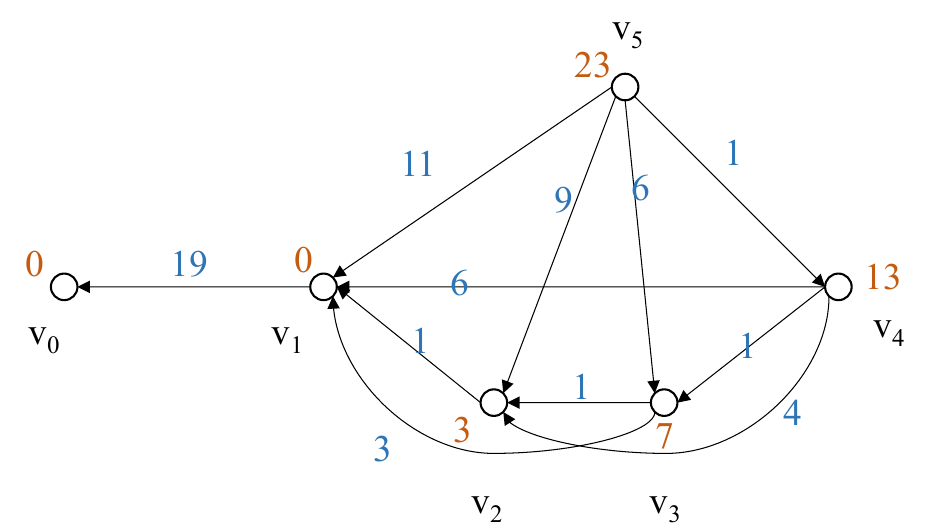}
    \caption{Martelli's instance for inconsistent A*. Blue integers are edge weights, and brown integers are heuristic values.}
    \label{fig:martelli}
\end{figure}

\noindent\textbf{Martelli's instances with exponential worst-case complexities}. Martelli's family of A* instances is based on a construction of negative-weight shortest path instances. Let $V=\{v_0,v_1,\dots,v_n\}$ denote the nodes of a directed graph, where $s=v_n$ is the start node and $t=v_0$ is the goal node. For any $1\le i<j\le n$, there is a directed edge from $v_j$ to $v_i$. We first define a negative-weight graph, where \begin{align*}
    w^\prime_{2,1}&=-2,\\
    w^\prime_{{i+1},1}&=w^\prime_{i,1}-(2^{i-2}+1),\,\forall\, 2\le i<n,\\
    w^\prime_{i,{j+1}}&=w^\prime_{i,j}+1,\,\forall\,1<j+1<i\le n.
\end{align*} We also add an extra edge, $$w^\prime_{1,0}=-\sum_{i=1}^{n-1}w^\prime_{{i+1},i}.$$
Then, we can apply the equivalence between negative-weight Dijkstra and A* to construct an A* instance. The heuristic values of nodes are defined by \begin{align*}
    h(v_0)&=h(v_1)=0,\\
    h(v_i)&=h(v_{i-1})+(2^{i-2}+2),\,\forall\,1<i\le n.
\end{align*} The weight of each edge $(i,j)$ is $$w_{i,j}=w^\prime_{i,j}+h(v_i)-h(v_j).$$ Figure~\ref{fig:martelli} gives the example of Martelli's construction for $n=5$.

\noindent\textbf{Random perturbations on Martelli's instances}. In Martelli's instance, the A* algorithm will enumerate every path from $s$ to other nodes, leading to $\Omega(2^n)$ time. We empirically show that with a little random perturbation, the total number of iterations greatly decreases.

The result is illustrated in Figure~\ref{fig:result}. We evaluate the A* algorithm on Martelli's instance for $n=15,20$ and $25$. (Larger instances are intractable due to the exponential time complexity.) We gradually increase the random perturbation from $\kappa = \infty$ (no perturbation). Each edge is independently perturbed by a uniform random variable supported on $[-\gamma,\gamma]$ for some $\gamma$ (and thus $1/\kappa$ is $\gamma$ divided by the maximum weight value). We clip the weights below zero such that all weights are non-negative. For each $\kappa$, we repeat the generation of the instances 100 times. Figure~\ref{fig:result} gives the average and the standard deviation of the total iteration number. Note that both axes are logarithmic. It is easy to notice that if there is no random perturbation, the iteration number is about $2^n$. As we slightly increase the perturbation, the iteration number greatly reduces. This verifies our theoretical results and indicates that the exponential behavior of inconsistent A* for worst-case instances is highly fragile.

\begin{figure}
    \centering
    \includegraphics[width=0.33\linewidth]{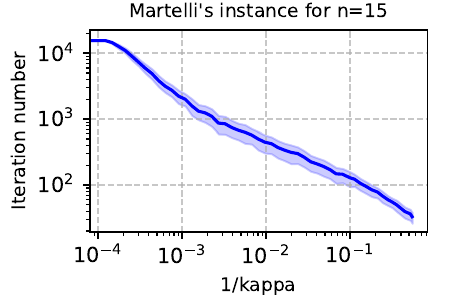}\includegraphics[width=0.33\linewidth]{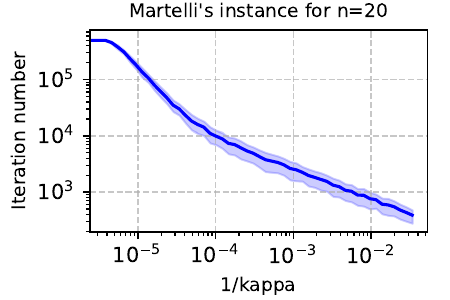}\includegraphics[width=0.33\linewidth]{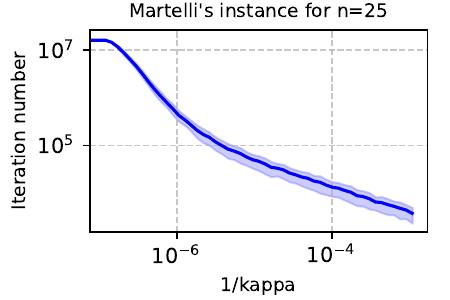}
    \caption{The iteration number of inconsistent A* on Martelli's instance.}
    \label{fig:result}
\end{figure}

\section{Conclusion}
In this paper, we study the smoothed complexity of the A* algorithm with inconsistent heuristics. We show that the exponential running time behavior of inconsistent A* in the worst case is highly fragile. With a little random perturbation, inconsistent A* finds the shortest path in polynomial time. Our result also works for the Dijkstra algorithm on negative-weight graphs. Numerical experiments verify our theoretical analysis.

There are two directions of future research. On the one hand, our bound is still loose compared with practice. Note that if the edge weights are integers in $[0,W]$, the time complexity of inconsistent A* is only $O(n^2W)$. We conjecture that a corresponding smoothed complexity of $O(n^2\kappa)$ is possible. On the other hand, it would be interesting to prove lower bounds for smoothed A* to see what kind of instances are really hard for inconsistent heuristics.

\bibliographystyle{unsrtnat}
\bibliography{edabib}

\appendix

\end{document}